\documentclass[onecolumn]{IEEEtran}
\usepackage[utf8]{inputenc}
\usepackage{amsmath, amssymb, amsthm}
\usepackage{authblk}
\usepackage{mathrsfs}
\usepackage{geometry}
\usepackage{hyperref}
\usepackage{xcolor}
\usepackage[switch]{lineno}
\usepackage{comment}
\newtheorem{theorem}{Theorem}[section]
\newtheorem{lemma}[theorem]{Lemma}

\newtheorem{corollary}[theorem]{Corollary}
\newtheorem{definition}[theorem]{Definition}
\newtheorem{remark}[theorem]{Remark}
\newtheorem{example}[theorem]{Example}

\newcommand{\Lb}{\mathbb{L}}
\newcommand{\Kb}{\mathbb{K}}
\newcommand{\Fb}{\mathbb{F}}
\newcommand{\Rcal}{\mathcal{R}}
\newcommand{\N}{\mathrm{N}}
\newcommand{\Tr}{\mathrm{Tr}}

\newcommand{\ev}{\operatorname{ev}}

\newcommand{\wt}{\mathrm{wt}}

\title{Twisted and Twisted Linearized Reed--Solomon Codes, LCD and ACD MDS constructions}

\author[1]{Sanjit Bhowmick}
\author[1]{Kuntal Deka}
\author[2]{Edgar Martínez-Moro\thanks{E. Martínez-Moro is partially supported by Grant SAGACT-1
MCIN/AEI/10.13039/501100011033 y FEDER ``Una manera de hacer Europa''
PID2022-138906NB-C21 (2023/2027)}}

\affil[1]{{\small Department of Electronics and Electrical Engineering\\
Indian Institute of Technology Guwahati\\
Assam, 781039, India.\\
E-mail: \rmfamily{sanjitbhowmick@rnd.iitg.ac.in; kuntaldeka@iitg.ac.in}}}

\affil[2]{{\small Institute of Mathematics, University of Valladolid,\\
Valladolid, Spain \\
E-mail: \rmfamily{edgar.martinez@uva.es}}}

\date{\today}

\begin{document}
\maketitle
\begin{abstract}
We investigate a natural subfamily of twisted linearized Reed--Solomon (TLRS) codes in the sum-rank metric, where the twist is applied only to the constant term. We establish a simple necessary and sufficient condition for these codes to be linear complementary dual (LCD): the twisting parameter \(\eta\) must satisfy \(\eta^2 \neq -1\) in the underlying field. This criterion is independent of the evaluation subgroup, the dimension parameter, and the twisting exponent (subject only to a mild restriction on the code length).
Furthermore, we construct infinite families of additive twisted linearized Reed--Solomon codes that are simultaneously additive complementary dual (ACD) and maximum distance separable (MDS) over quadratic extensions \(\mathbb{F}_{q^2}\), with respect to the trace-Hermitian inner product. These codes are explicit and achieve optimal parameters for all admissible lengths.
\end{abstract}


\section{Introduction}

Twisted Reed--Solomon (TRS) codes were introduced in~\cite{Beelen,9691371} as evaluation codes in the Hamming metric that achieve the maximum distance separable (MDS) property. The construction was inspired by twisted Gabidulin codes~\cite{Sheekey}, which are known to be maximum-rank-distance (MRD) codes but are not equivalent to classical Gabidulin codes (the rank-metric analogs of Reed--Solomon codes).
Linearized Reed--Solomon (LRS) codes have recently emerged as a powerful generalization that unifies classical Reed--Solomon and Gabidulin codes within the sum-rank metric framework~\cite{Liu2025LRS}. Defined as evaluation codes over skew polynomial rings, LRS codes attain the Singleton bound in the sum-rank metric and are therefore maximum sum-rank distance (MSRD) codes. Since the sum-rank metric naturally interpolates between the Hamming and rank metrics, LRS codes provide a unified algebraic approach that is particularly well-suited for applications such as multishot network coding, distributed storage, and space-time coding.

A key practical motivation in modern coding theory is the design of codes with support-constrained generator matrices, i.e., matrices in which certain entries must be zero. Such constraints arise naturally in distributed and networked systems, for instance in wireless cooperative data exchange and multi-source network coding, where encoding operations are restricted by locality or the underlying communication topology~\cite{Liu2025LRS}. 
In the classical Hamming and rank-metric settings, the existence of optimal codes under support constraints is governed by the well-known GM-MDS and GM-MRD conditions, respectively. Extending this line of research to the sum-rank metric, recent results establish analogous combinatorial conditions that guarantee the existence of maximum sum-rank distance (MSRD) codes. In particular, LRS codes provide an explicit algebraic framework for constructing support-constrained MSRD codes under suitable field-size conditions, thereby generalizing the corresponding results for Reed--Solomon and Gabidulin codes~\cite{Liu2025LRS}.

Beyond code construction, substantial progress has also been achieved in decoding algorithms for LRS codes. In particular, syndrome-based error-erasure decoding algorithms for interleaved LRS codes have been developed that efficiently correct both errors and erasures in the sum-rank metric~\cite{Hormann2026}. These methods exploit structured interleaving, including both vertical and horizontal schemes, to significantly improve decoding performance against burst errors and correlated failures. The resulting decoders offer strong probabilistic performance guarantees and naturally extend classical decoding techniques (such as those for Reed--Solomon and Gabidulin codes) to the sum-rank metric, thereby enabling robust and efficient recovery in applications such as multishot network coding and distributed storage systems. 
Together, these advances demonstrate that LRS codes constitute a versatile and powerful framework that bridges theory and practice in modern coding theory. Their applications range from distributed communication systems to secure and reliable data transmission.

On the other hand, the theory of codes in the sum-rank metric has seen significant development in recent years, largely due to the foundational contributions of Martínez-Peñas~\cite{Umberto} and subsequent works. These codes naturally generalize both classical Hamming-metric codes and rank-metric codes. They arise in a variety of applications, including multishot network coding, distributed storage systems, and space-time coding.
Twisted linearized Reed--Solomon (TLRS) codes were introduced in~\cite{Neri2022b} as a broad family of MSRD codes. Their construction builds upon the algebraic framework developed for sum-rank metric codes. More precisely, TLRS codes are realized as evaluations of certain subsets of a quotient algebra of \(\theta\)-skew polynomials. This approach provides an alternative interpretation of the sum-rank metric and, as shown in~\cite{Neri2022b}, is isometric to the standard sum-rank metric framework.

Linear complementary dual (LCD) codes are linear codes \(\mathcal{C}\) satisfying \(\mathcal{C} \cap \mathcal{C}^\perp = \{\mathbf{0}\}\). They were introduced by Massey~\cite{Massey1992} and have found important applications in cryptography, particularly for countering side-channel and fault-injection attacks~\cite{Carlet2016}. Subsequent work by Jin~\cite{Jin2017} constructed new families of MDS LCD codes using generalized Reed--Solomon codes. In recent years, LCD codes and their variants including linear complementary pairs and codes with small hulls  have received considerable attention~\cite{CMTYP2018,CMTQ2019,Chen2023,LSEL2024,Sok2022}.
While LCD codes have been extensively studied in the Hamming and rank metrics (see, e.g.,~\cite{Survey,Liu2019,liang2024} and references therein), their theory in the sum-rank metric remains largely unexplored. To the best of the authors' knowledge, combining the LCD property with the sum-rank metric is still in its early stages, yet it holds strong potential for enhancing security in networked and distributed systems. 
On the other hand, additive complementary dual (ACD) codes, recently introduced and studied in~\cite{Shi2022,Choi2023}, generalize LCD codes by considering \(\mathbb{F}_q\)-linear codes over extension fields \(\mathbb{F}_{q^m}\). This extension offers greater flexibility in parameter selection and opens new avenues for code design.

In the Hamming metric, TRS codes have been successfully employed to construct several families of LCD MDS codes. Liu \textit{et al.}~\cite{liu2021} presented two explicit constructions of MDS twisted Reed--Solomon codes and derived LCD MDS codes from self-orthogonal codes. Liang \textit{et al.}~\cite{liang2025} systematically studied \((\mathcal{L},\mathcal{P})\)-twisted generalized Reed--Solomon codes, establishing parity-check matrices, self-orthogonality conditions, and a complete characterization of near-MDS codes. Huang \textit{et al.}~\cite{huang2023} constructed three additional families of LCD MDS codes from twisted Reed--Solomon codes by exploiting factorizations of binomials, trinomials, and linearized polynomials. Finally, Liang \textit{et al.}~\cite{liang2024} unified and extended these results by introducing \((*)\)- \((\mathcal{L},\mathcal{P})\)-twisted generalized Reed--Solomon codes, thereby obtaining four broad classes of LCD codes.

 In this paper, we focus on a natural subfamily of TLRS codes (see Definition~\ref{def:family}) in which only the constant term is twisted. For this family, we provide a simple and precise condition guaranteeing the LCD property (Theorem~\ref{thm:main}): the code is LCD if and only if \(\eta^2 \neq -1\) in the underlying field, where \(\eta\) appears in the definition of the twist.
All previously known constructions of LCD MDS codes from twisted Reed--Solomon codes operate exclusively in the Hamming metric. They rely on classical polynomial rings \(\mathbb{F}_q[x]\) and use tools such as Vandermonde matrices, power sums, and elementary linear algebra. In contrast, our approach works in the sum-rank metric and is based on the theory of skew polynomial rings developed in~\cite{Neri2022b}, which is fundamentally different. The resulting LCD criterion is universal, independent of the evaluation subgroup and the code parameters which exhibits a level of simplicity not observed in earlier works.

In the second part of the paper (Section~\ref{sec:acd}), we consider the trace-Hermitian inner product and construct ACD MSRD codes over \(\mathbb{F}_{q^2}\) for all prime powers \(q \equiv 1 \pmod{4}\) with \(q \ge 5\). To obtain an additive twisted Reed--Solomon code, we restrict the coefficients of skew polynomials such that the constant term and the coefficient of \(X^k\) lie in \(\mathbb{F}_q\), while the remaining coefficients lie in \(\mathbb{F}_{q^2}\). This yields \(\mathbb{F}_q\)-dimension \(2k\) and achieves the Singleton bound \(d = \ell - k + 1\) in the sum-rank metric. To the best of our knowledge, no prior work has used twisted constructions to obtain ACD codes over extension fields. Our construction provides the first explicit family of ACD MSRD codes over \(\mathbb{F}_{q^2}\) for all such \(q\), with flexible lengths \(\ell \ge 2k\). It relies on a suitable element \(\gamma = \alpha\) satisfying \(\alpha^q = -\alpha\) together with a probabilistic argument to guarantee the ACD property.

The paper is organized as follows. Section~\ref{sec:pre} provides the necessary background on skew polynomial rings, the sum-rank metric, and the associated evaluation map. In Section~\ref{sec:lcd}, we introduce the family of TLRS codes and establish the main LCD characterization (Theorem~\ref{thm:main}). Section~\ref{sec:acd} is devoted to ACD codes: we present an explicit construction of ACD MSRD codes over \(\mathbb{F}_{q^2}\), prove that these codes achieve the MDS property (Corollary~\ref{cor:mds_sufficient}), and derive necessary and sufficient conditions for the ACD property (Theorem~\ref{thm:acd_condition}). Finally, we prove the existence of such codes for all admissible parameters (Theorem~\ref{thm:main_acd}).

\section{Preliminaries}\label{sec:pre}
In this section, we  introduce the basic facts on TLRS codes.  For a detailed treatment of them within the skew polynomial framework and their relationship to rank-metric codes, the reader is referred to \cite{Neri2022b}.

Throughout this paper, let \(\mathbb{K} = \mathbb{F}_q\) denote a finite field with \(q\) elements, and let \(\mathbb{L} = \mathbb{F}_{q^r}\) be the unique (up to isomorphism) Galois extension of \(\mathbb{K}\) of degree \(r\). The Frobenius automorphism generating \(\mathrm{Gal}(\mathbb{L}|\mathbb{K})\) is denoted by \(\theta \colon x \mapsto x^q\).
Let \(\Lambda \subseteq \mathbb{K}^*\) be the unique multiplicative subgroup of order \(\ell\), where \(\ell\) divides \(q-1\). Thus, $\ell$ is an invertible element  in the ring $\mathbb{Z}_q$. Since \(\Lambda\) is cyclic, we fix a generator \(\gamma \in \Lambda\) and label its elements as
$
\lambda_i = \gamma^{i-1} \quad \text{for} \quad i = 1, \dots, \ell.
$
We denote the norm map of the extension \(\mathbb{L}|\mathbb{K}\) by \(\mathrm{N}_{\mathbb{L}|\mathbb{K}}\). Because the norm is surjective onto \(\mathbb{K}^*\), for each \(i = 1, \dots, \ell\) there exists an element \(\alpha_i \in \mathbb{L}\) such that
$
\mathrm{N}_{\mathbb{L}|\mathbb{K}}(\alpha_i) = \lambda_i.
$

The \(\theta\)-skew polynomial ring \(\mathbb{L}[X;\theta]\) consists of all polynomials \(h(X) = \sum_{i=0}^n a_i X^i\) with \(a_i \in \mathbb{L}\) and \(n \ge 0\), equipped with the usual addition and the multiplication rule
$
X a = \theta(a) X
$
for all \(a \in \mathbb{L}\), extended to arbitrary polynomials by distributivity.
Define the skew polynomial
\[
H_\Lambda(X) = \prod_{i=1}^\ell (X^r - \lambda_i) \in \mathbb{L}[X;\theta].
\]
The quotient algebra \(\mathcal{R}_\Lambda = \mathbb{L}[X;\theta] / (H_\Lambda(X))\) is isomorphic, as a \(\mathbb{K}\)-algebra, to \((\mathbb{L}[\theta])^\ell\) via the evaluation map
\begin{equation}\label{eq:eval}
\Phi_\alpha \colon \mathcal{R}_\Lambda \longrightarrow (\mathbb{L}[\theta])^\ell, \qquad
\Phi_\alpha(F) = \bigl(F(\alpha_1),\dots,F(\alpha_\ell)\bigr).
\end{equation}
Here, the evaluation at each \(\alpha_i\) is defined by the rule
$
X^j(\alpha_i) = \lambda_i^j \theta^j,
$
where \(\theta^j\) denotes the \(j\)-fold composition of the Frobenius automorphism \(\theta\), and the right-hand side is viewed as an element of the group algebra \(\mathbb{L}[\theta]\). The map \(\Phi_\alpha\) is a \(\mathbb{K}\)-algebra isomorphism (see \cite[Theorem 4.1]{Neri2022b}).
Also, \(\mathbb{L}[\theta]\) is an \(r\)-dimensional vector space over \(\mathbb{K}\) with basis \(\{1, \theta, \dots, \theta^{r-1}\}\). Its elements are polynomials in \(\theta\) with coefficients in \(\mathbb{L}\) and are called \(\theta\)-polynomials (or \(\theta\)-linearized polynomials).

\begin{definition}
A \emph{sum-rank metric code} is any \(\mathbb{K}\)-linear subspace \(C \subseteq \mathcal{R}_\Lambda\).
\end{definition}

\begin{definition}[Twisted linearized Reed--Solomon codes]\label{def:family}
Let \(1 \le k \le \ell r - 1\), \(0 \le h \le r-1\), and \(\eta \in \mathbb{L}^*\). The \emph{twisted linearized Reed--Solomon (TLRS) code} with parameters \((k,h,\eta)\) is the \(\mathbb{K}\)-linear subspace
\[
\mathcal{L}_k^\theta(\eta,h) = \Bigl\{ f_0 + f_1 X + \cdots + f_{k-1} X^{k-1} + \eta\,\theta^h(f_0)\,X^k \;\Bigm|\; f_i \in \mathbb{L} \Bigr\} \subseteq \mathcal{R}_\Lambda.
\]
\end{definition}

Note that \(\mathcal{L}_k^\theta(\eta,h)\) is a \(\mathbb{K}\)-vector space of dimension \(kr\). The first \(k\) coefficients \(f_0,\dots,f_{k-1}\) are free, and the twist term \(\eta\,\theta^h(f_0)X^k\) introduces no additional linear dependence over \(\mathbb{K}\).

Let \(\operatorname{Tr}_{\mathbb{L}|\mathbb{K}}\) denote the trace map of the extension \(\mathbb{L}|\mathbb{K}\). The standard \(\mathbb{K}\)-bilinear form on \(\mathcal{R}_\Lambda\) is defined by
\[
\langle F,G\rangle_\Lambda = \operatorname{Tr}_{\mathbb{L}|\mathbb{K}}\!\Bigl( \sum_{i=0}^{\ell r-1} f_i g_i \Bigr),
\]
where \(F = \sum f_i X^i\) and \(G = \sum g_i X^i\) are the unique representatives of degree less than \(\ell r\) in the residue classes \(F + (H_\Lambda)\) and \(G + (H_\Lambda)\), respectively.

Under the evaluation isomorphism~\eqref{eq:eval}, this form takes the explicit shape (see~\cite[Proposition~5.3]{Neri2022b})
\begin{equation}
\langle F,G\rangle_\Lambda = \ell^{-1} \sum_{i=1}^\ell \operatorname{Tr}_{\mathbb{L}|\mathbb{K}}\bigl( F(\alpha_i)\, G(\alpha_i^{-1}) \bigr).
\end{equation}
Note that we abuse the notation and call $\ell$ the image of the mapping of $\ell$ by the canonical ring homomorphism  $\mathbb Z\mapsto \Kb$ which 
sends an integer $m$ to $m \cdot 1_\Kb$, i.e., the sum of $m$ copies of the multiplicative
identity. Since $\Kb$ has the characteristic $p$, we have that the kernel is $\mathbb F_p$, thus the image of $\ell$ is  $\ell \mod p$ in $\mathbb F_p\subset \Kb$, and it is an invertible element.

\begin{definition}
Let \(C \subseteq \mathcal{R}_\Lambda\) be a \(\mathbb{K}\)-linear sum-rank metric code. Its \emph{dual code} with respect to \(\langle\,\cdot\,,\,\cdot\,\rangle_\Lambda\) is
\[
C^\perp = \bigl\{ G \in \mathcal{R}_\Lambda \;\bigm|\; \langle F,G\rangle_\Lambda = 0 \text{ for all } F \in C \bigr\}.
\]
\end{definition}
\section{Characterization of linear complementary dual codes}\label{sec:lcd}

 Given a fixed \(\mathbb{K}\)-basis \(\{\beta_1, \dots, \beta_r\}\) of \(\mathbb{L}\), a natural \(\mathbb{K}\)-basis for the TLRS code \(\mathcal{C} = \mathcal{L}_k^\theta(\eta,h)\) is given by
\begin{equation}\label{eq:basis}
\mathcal{B} = \bigl\{ \beta_t X^j \;\bigm|\; t=1,\dots,r,\; j=1,\dots,k-1 \bigr\} 
\;\cup\; 
\bigl\{ \beta_t + \eta\,\theta^h(\beta_t) X^k \;\bigm|\; t=1,\dots,r \bigr\}.
\end{equation}
When \(k=1\), the first set is empty and \(\mathcal{B}\) reduces to \(\{ \beta_t + \eta\,\theta^h(\beta_t) X \mid t=1,\dots,r \}\).
Note that clearly, this set contains exactly \(r(k-1) + r = kr\) elements and forms a \(\mathbb{K}\)-basis of \(\mathcal{C}\).
We will use the following technical lemma to compute the Gram matrix of \(\mathcal{C}\) with respect to the basis \(\mathcal{B}\) in~\eqref{eq:basis}.

\begin{lemma}\label{lem:trace}
Let \(\beta_t, \beta_u \in \mathbb{L}\) and \(h \in \mathbb{Z}\). Then
\[
\operatorname{Tr}_{\mathbb{L}|\mathbb{K}}\bigl( \theta^h(\beta_t) \cdot \theta^h(\beta_u) \bigr) = \operatorname{Tr}_{\mathbb{L}|\mathbb{K}}(\beta_t \cdot \beta_u).
\]
\end{lemma}
\begin{proof}
Since \(\theta^h\) is a field automorphism of \(\mathbb{L}\), we have that
$
\theta^h(\beta_t) \cdot \theta^h(\beta_u) = \theta^h(\beta_t \cdot \beta_u)$.
The trace \(\operatorname{Tr}_{\mathbb{L}|\mathbb{K}}\) is the sum of all \(\mathbb{K}\)-embeddings of \(\mathbb{L}\) into an algebraic closure. Applying the automorphism \(\theta^h\) merely permutes these embeddings. Therefore, the value of the trace remains invariant $
\operatorname{Tr}_{\mathbb{L}|\mathbb{K}}\bigl( \theta^h(\beta_t \cdot \beta_u) \bigr) = \operatorname{Tr}_{\mathbb{L}|\mathbb{K}}(\beta_t \cdot \beta_u)$.
\end{proof}

The Gram matrix \(G\) of \(\mathcal{C} = \mathcal{L}_k^\theta(\eta,h)\) with respect to the basis \(\mathcal{B}\) in~\eqref{eq:basis} is determined by the following three types of inner products.
\begin{description}
 \item[\textbf{1.-}] Two elements from the first part of \(\mathcal{B}\), i.e., \(\beta_t X^j\) and \(\beta_u X^{j'}\) with \(1 \le j,j' \le k-1\).
The inner product is nonzero  when \(j = j'\), in that case
\[
\langle \beta_t X^j, \beta_u X^{j'} \rangle_\Lambda = \delta_{j,j'} \operatorname{Tr}_{\mathbb{L}/\mathbb{K}}(\beta_t \beta_u).
\]

\item[\textbf{2.-}] One element from each part of the set \(\mathcal{B}\), i.e., \(\beta_t X^j\) (\(1 \le j \le k-1\)) and \(\beta_u + \eta\theta^h(\beta_u)X^k\).
Since the supports of the coefficients are disjoint (the first vector has no \(X^0\) or \(X^k\) term when \(j \ge 1\), and the second has zero coefficient in degrees \(1\) to \(k-1\)), we have
\[
\langle \beta_t X^j,\; \beta_u + \eta\theta^h(\beta_u)X^k \rangle_\Lambda = 0.
\]

\item[\textbf{3.-}] Two elements from the second part of the set \(\mathcal{B}\), i.e.,
\(\beta_t + \eta\theta^h(\beta_t)X^k\) and \(\beta_u + \eta\theta^h(\beta_u)X^k\).
The only nonzero contributions come from degrees \(0\) and \(k\):
\[
\langle \beta_t + \eta\theta^h(\beta_t)X^k,\; \beta_u + \eta\theta^h(\beta_u)X^k \rangle_\Lambda
= \operatorname{Tr}_{\mathbb{L}|\mathbb{K}}(\beta_t\beta_u) + \operatorname{Tr}_{\mathbb{L}|\mathbb{K}}\bigl(\eta^2 \theta^h(\beta_t)\theta^h(\beta_u)\bigr).
\]
Using Lemma~\ref{lem:trace} and the \(\mathbb{K}\)-linearity of the trace, this simplifies to
\begin{align*}
\operatorname{Tr}_{\mathbb{L}|\mathbb{K}}(\beta_t\beta_u) + \operatorname{Tr}_{\mathbb{L}|\mathbb{K}}\bigl(\eta^2 \theta^h(\beta_t)\theta^h(\beta_u)\bigr)
&= \operatorname{Tr}_{\mathbb{L}|\mathbb{K}}\bigl( \theta^h(\beta_t)\theta^h(\beta_u) \bigr) + \operatorname{Tr}_{\mathbb{L}|\mathbb{K}}\bigl(\eta^2 \theta^h(\beta_t)\theta^h(\beta_u)\bigr) \\
&= \operatorname{Tr}_{\mathbb{L}|\mathbb{K}}\bigl( (1 + \eta^2) \theta^h(\beta_t)\theta^h(\beta_u) \bigr).
\end{align*}
Applying the automorphism \(\theta^{-h}\) (which leaves the trace invariant) yields
\[
\operatorname{Tr}_{\mathbb{L}|\mathbb{K}}\bigl( (1 + \eta^2) \theta^h(\beta_t)\theta^h(\beta_u) \bigr)
= \operatorname{Tr}_{\mathbb{L}|\mathbb{K}}\bigl( \theta^{-h}(1 + \eta^2) \cdot \beta_t \beta_u \bigr).
\]
Let \(\alpha := \theta^{-h}(1 + \eta^2) \in \mathbb{L}\). Then the \(r \times r\) block of the Gram matrix corresponding to the second part of the basis is
$
B = \bigl[ \operatorname{Tr}_{\mathbb{L}/\mathbb{K}}(\alpha \cdot \beta_t \beta_u) \bigr]_{t,u=1}^r$.
\end{description}

From the preceding calculations, if we let \(M\) be the \(r \times r\) matrix with entries
$
M_{t,u} = \operatorname{Tr}_{\mathbb{L}|\mathbb{K}}(\beta_t \beta_u)$, where  $t,u=1,\dots,r,
$,
then the Gram matrix \(G\) of \(\mathcal{C}\) with respect to the basis \(\mathcal{B}\) takes the block-diagonal form
\[
G = \begin{pmatrix}
I_{k-1} \otimes M & 0 \\[2mm]
0 & B
\end{pmatrix},
\]
where \(I_{k-1}\) denotes the \((k-1) \times (k-1)\) identity matrix (this block is absent when \(k=1\)).

Since the trace form is non-degenerate on \(\mathbb{L}|\mathbb{K}\), we have \(\det M \neq 0\) in \(\mathbb{K}\). Therefore,
\[
\det G = (\det M)^{k-1} \cdot \det B.
\]

The matrix \(B\) is the Gram matrix of the \(\mathbb{K}\)-bilinear form \((x,y) \mapsto \operatorname{Tr}_{\mathbb{L}|\mathbb{K}}(\alpha \, x y)\) with respect to the basis \(\{\beta_1,\dots,\beta_r\}\). This form is non-degenerate if and only if \(\alpha \neq 0\), because multiplication by a nonzero element of \(\mathbb{L}\) is a \(\mathbb{K}\)-linear isomorphism and the trace form itself is non-degenerate. Consequently,
$
\det B = 0 $ if and only if $\alpha = 0$, and
recalling that \(\alpha = \theta^{-h}(1 + \eta^2)\), we obtain
$
\det G = 0$ if and only if $\eta^2 = -1 \in\mathbb{L}$.
Thus, the Gram matrix \(G\) is singular if and only if \(\eta^2 = -1\) in the extension field \(\mathbb{L}\).

\begin{definition}\label{def:LCD}
A \(\mathbb{K}\)-linear sum-rank metric code \(\mathcal{C} \subseteq \mathcal{R}_\Lambda\) is called a \emph{linear complementary dual (LCD) code} if
$
\mathcal{C} \cap \mathcal{C}^\perp = \{\mathbf{0}\}
$ and $
\mathcal{C} \cup \mathcal{C}^\perp = \mathcal{R}_\Lambda
$.
\end{definition}

Note that when \(\mathcal{C}\) is a TLRS code of \(\mathbb{K}\)-dimension \(kr\), its dual \(\mathcal{C}^\perp\) has dimension \(\ell r - kr\) by~\cite[Proposition~6.11]{Neri2022b}. Therefore, it suffices to verify the intersection condition \(\mathcal{C} \cap \mathcal{C}^\perp = \{\mathbf{0}\}\) (as the direct-sum condition then follows automatically by dimension counting).
The following result completely characterizes when a TLRS code is LCD.

\begin{theorem}[LCD Characterization]\label{thm:main}
Let \(\mathcal{C} = \mathcal{L}_k^\theta(\eta,h)\) be a twisted linearized Reed--Solomon code in \(\mathcal{R}_\Lambda\) with parameters \(1 \le k \le \ell r - 1\) and \(\eta \in \mathbb{L}^*\). Then \(\mathcal{C}\) is an LCD code if and only if
$
1 + \eta^2 \neq 0 \quad \text{in } \mathbb{L}.
$
\end{theorem}

\begin{proof}
Since \(\dim_{\mathbb{K}} \mathcal{C} + \dim_{\mathbb{K}} \mathcal{C}^\perp = \dim_{\mathbb{K}} \mathcal{R}_\Lambda\), the code \(\mathcal{C}\) is LCD if and only if the restriction of the bilinear form \(\langle \cdot, \cdot \rangle_\Lambda\) to \(\mathcal{C}\) is non-degenerate.  { Indeed, if there exists a non‑zero $c\in \mathcal{C}$ with $\langle c,c'\rangle=0$ for all $c'\in \mathcal{C}$, then $c\in \mathcal{C}^\perp$, so $c\in \mathcal{C}\cap \mathcal{C}^\perp$. Conversely, if $0\neq c\in \mathcal{C}\cap \mathcal{C}^\perp$, then $\langle c,c'\rangle=0$ for every $c'\in \mathcal{C}$, so the restriction is degenerate.} This fact is equivalent to the Gram matrix of any $\Kb$-basis of $\mathcal{C}$ being invertible. Using the basis $\mathcal{B}$ described above, the Gram matrix $G$ satisfies
\[
\det G = 0 \quad \Longleftrightarrow \quad 1 + \eta^2 = 0 \quad \text{in } \mathbb{L}.
\]
Therefore, \(\mathcal{C} \cap \mathcal{C}^\perp \neq \{\mathbf{0}\}\) if and only if \(1 + \eta^2 = 0\), which completes the proof.
\end{proof}

\begin{remark}
The twisting exponent \(h\) does not appear in the final condition, as it is cancelled by the application of \(\theta^{-h}\). Moreover, the characterization is independent of the choice of the subgroup \(\Lambda\) (provided \(\ell\) is invertible in \(\mathbb{K}\)) and of the dimension parameter \(k\), as long as \(1 \le k \le \ell r - 1\). The boundary case \(k = \ell r\) would require reduction modulo \(H_\Lambda(X)\) and leads to a more involved analysis; we therefore restrict ourselves to \(k \le \ell r - 1\).
\end{remark}

We illustrate the theorem with explicit computations over the quadratic extension \(\mathbb{F}_{25}|\mathbb{F}_5\).

\begin{example}[LCD code]\label{ex:LCD}
Let \(\mathbb{K} = \mathbb{F}_5\) and \(\mathbb{L} = \mathbb{F}_{25} = \mathbb{F}_5[u]\) with \(u^2 = 2\). The Frobenius automorphism is \(\theta(x) = x^5\), satisfying \(\theta(u) = 4u\). Take \(\Lambda = \{1,4\} \subseteq \mathbb{F}_5^*\) (order \(\ell=2\)), \(k=1\), \(h=0\), and \(\eta = 2 + u\).

Then
\[
\eta^2 = (2+u)^2 = 4 + 4u + u^2 = 4 + 4u + 2 = 6 + 4u \equiv 1 + 4u \pmod{5}.
\]
Hence \(1 + \eta^2 = 2 + 4u \neq 0\) in \(\mathbb{F}_{25}\). By Theorem~\ref{thm:main}, \(\mathcal{C} = \mathcal{L}_1^\theta(\eta,0)\) is an LCD code.
To verify this fact directly, note that \(\mathcal{C} = \{ f_0 + \eta f_0 X \mid f_0 \in \mathbb{L} \}\). With respect to the \(\mathbb{F}_5\)-basis \(\{1,u\}\) of \(\mathbb{L}\), a basis of \(\mathcal{C}\) is
$
\{1 + \eta X,\; u + u\eta X\}$.
The Gram matrix with respect to this basis is
\[
G = \begin{pmatrix}
\operatorname{Tr}_{\mathbb{F}_{25}|\mathbb{F}_5}(1+\eta^2) & \operatorname{Tr}_{\mathbb{F}_{25}|\mathbb{F}_5}(u(1+\eta^2)) \\
\operatorname{Tr}_{\mathbb{F}_{25}|\mathbb{F}_5}(u(1+\eta^2)) & \operatorname{Tr}_{\mathbb{F}_{25}|\mathbb{F}_5}(u^2(1+\eta^2))
\end{pmatrix}.
\] Substituting \(1 + \eta^2 = 2 + 4u\) and using the known values \(\operatorname{Tr}_{\mathbb{F}_{25}|\mathbb{F}_5}(1) = 2\), \(\operatorname{Tr}_{\mathbb{F}_{25}|\mathbb{F}_5}(u) = 0\), \(\operatorname{Tr}_{\mathbb{F}_{25}|\mathbb{F}_5}(u^2) = \operatorname{Tr}_{\mathbb{F}_{25}|\mathbb{F}_5}(2) = 4\), direct computation yields
\[
G \equiv \begin{pmatrix} 4 & 1 \\ 1 & 3 \end{pmatrix} \pmod{5}.
\]
Since \(\det G = 12 - 1 = 11 \equiv 1 \pmod{5} \neq 0\), the form is non-degenerate and \(\mathcal{C}\) is indeed LCD.
\end{example}

\begin{example}[Non-LCD code]
Keep the same field extension and parameters as in Example~\ref{ex:LCD}, but choose \(\eta = 2\). Then \(\eta^2 = 4 \equiv -1 \pmod{5}\), so \(1 + \eta^2 = 0\). By Theorem~\ref{thm:main}, \(\mathcal{C} = \mathcal{L}_1^\theta(2,0)\) is not an LCD code.
Direct verification shows that the Gram matrix with respect to the basis \(\{1 + 2X,\; u + 2u X\}\) is the zero matrix, since every entry contains the factor \(1 + \eta^2 = 0\). Thus \(\mathcal{C}\) is self-orthogonal (\(\mathcal{C} \subseteq \mathcal{C}^\perp\)) and hence not LCD.
\end{example}

\section{Additive complementary dual codes}\label{sec:acd}

Throughout this section, the non-zero integer \(q\) will denote a prime power satisfying \(q \equiv 1 \pmod{4}\) and \(q \ge 5\). 
For an \(\mathbb{F}_q\)-linear code \(\mathcal{C} \subseteq \mathbb{F}_{q^2}^n\), we use the notation \([n, q^k, d]_{q^2}\), where \(k = \dim_{\mathbb{F}_q} \mathcal{C}\) is the dimension over \(\mathbb{F}_q\) and \(d\) is the minimum Hamming distance of \(\mathcal{C}\).
The \emph{trace-Hermitian inner product} on \(\mathbb{F}_{q^2}^n\) is defined by
\begin{equation}
\langle \mathbf{u}, \mathbf{v} \rangle_H = \operatorname{Tr}_{\mathbb{F}_{q^2}|\mathbb{F}_q}\!\left( \sum_{i=1}^n u_i v_i^q \right)
\end{equation}
for all \(\mathbf{u} = (u_1,\dots,u_n), \mathbf{v} = (v_1,\dots,v_n) \in \mathbb{F}_{q^2}^n\).

\begin{definition}
An \(\mathbb{F}_q\)-linear code \(\mathcal{C} \subseteq \mathbb{F}_{q^2}^n\) is called an ACD code if
$
\mathcal{C} \cap \mathcal{C}^\perp = \{\mathbf{0}\}$, 
where \(\mathcal{C}^\perp\) denotes the dual of \(\mathcal{C}\) with respect to the trace-Hermitian inner product.
\end{definition}

\begin{lemma}[{\cite[Theorem 2.8]{Choi2023}}]\label{thm:acd_char}
Let \(\mathcal{C} \subseteq \mathbb{F}_{q^m}^n\) be an \(\mathbb{F}_q\)-linear code with generator matrix \(G \in \mathbb{F}_{q^m}^{k \times n}\). Then \(\mathcal{C}\) is an ACD code if and only if the matrix \(\mathfrak{T}(G G^\dagger)\) is invertible over \(\mathbb{F}_q\), where \((G^\dagger)_{i,j} = (G_{j,i})^q\) and \(\mathfrak{T}(M)_{i,j} = \operatorname{Tr}_{\mathbb{F}_{q^m}|\mathbb{F}_q}(M_{i,j})\).
\end{lemma}

When an \(\mathbb{F}_q\)-linear code \(\mathcal{C} \subseteq \mathbb{F}_{q^2}^n\) has even \(\mathbb{F}_q\)-dimension, i.e. \(\dim_{\mathbb{F}_q} \mathcal{C} = k = 2k'\) (see~\cite[Section 3.2]{Choi2023}), it satisfies the following \emph{Singleton bound}
\begin{equation}\label{eq:sing}
d \le n - k' + 1,
\end{equation}
where \(d\) is the minimum Hamming distance of \(\mathcal{C}\).

\begin{definition}
An additive code \(\mathcal{C}:=(n, q^{2k'}, d)\) over \(\mathbb{F}_{q^2}\) is called MDS if it attains the Singleton bound~\eqref{eq:sing} with equality, i.e., \(d = n - k' + 1\).
\end{definition}

Note that we may choose an \(\mathbb{F}_q\)-basis \(\{1, \alpha\}\) of \(\mathbb{F}_{q^2}\) such that \(\alpha^q = -\alpha\). Such an element \(\alpha \in \mathbb{F}_{q^2} \setminus \mathbb{F}_q\) exists because \(\mathbb{F}_{q^2}^*\) contains an element of order \(2(q-1)\). For any such element, \(\alpha^{q-1} = -1\), and therefore
\[
\alpha^q = \alpha \cdot \alpha^{q-1} = \alpha \cdot (-1) = -\alpha.
\]
From \(\alpha^q = -\alpha\) it follows that
\[
(\alpha^2)^q = \alpha^{2q} = (-\alpha)^2 = \alpha^2,
\]
so \(\alpha^2 \in \mathbb{F}_q\). Moreover, since \(q \equiv 1 \pmod{4}\), \(-1\) is a square in \(\mathbb{F}_q\). The element \(\alpha^2\) is a nonsquare in \(\mathbb{F}_q\): if \(\alpha^2 = \beta^2\) for some \(\beta \in \mathbb{F}_q^*\), then \((\alpha/\beta)^2 = 1\), so \(\alpha/\beta = \pm 1\), which would imply \(\alpha \in \mathbb{F}_q\), a contradiction. Consequently, \(-\alpha^2\) is also a nonsquare in \(\mathbb{F}_q\).

\begin{definition}\label{def:twisted}
Let \(\Lambda = \{\lambda_1, \dots, \lambda_\ell\} \subset \mathbb{F}_q^*\) be a set of \(\ell\) distinct elements. Define the evaluation map
\[
\ev_\Lambda \colon \mathbb{F}_{q^2}[X] \to \mathbb{F}_{q^2}^\ell, \qquad F \mapsto \bigl( F(\lambda_1), \dots, F(\lambda_\ell) \bigr).
\]
For integers \(1 \le k \le \ell-1\) and \(\gamma \in \mathbb{F}_{q^2}^*\), let
\[
\mathcal{D}_k(\gamma) = \Bigl\{ a_0 + \sum_{i=1}^{k-1} a_i X^i + \gamma a_k X^k \;\Bigm|\; a_0,a_k \in \mathbb{F}_q,\ a_1,\dots,a_{k-1} \in \mathbb{F}_{q^2} \Bigr\}.
\]
The image
$
\mathcal{C}_k(\gamma) = \ev_\Lambda\bigl( \mathcal{D}_k(\gamma) \bigr)
$ 
is called the \emph{additive TRS code} with parameters \((k, \gamma, \Lambda)\).
\end{definition}

\begin{lemma}
With the notation in Definition~\ref{def:twisted},
$
\dim_{\mathbb{F}_q} \mathcal{C}_k(\gamma) = 2k$.
\end{lemma}

\begin{proof}
The evaluation map \(\ev_\Lambda\) is injective on the set of polynomials of degree at most \(k\), since \(|\Lambda| = \ell > k\). Therefore, it is injective when restricted to \(\mathcal{D}_k(\gamma)\). As \(\mathcal{D}_k(\gamma)\) is an \(\mathbb{F}_q\)-vector space of dimension \(2k\) (one free parameter in \(\mathbb{F}_q\) for \(a_0\) and \(a_k\), and two for each of the \(k-1\) middle coefficients), the claim follows.
\end{proof}

\begin{lemma}\label{lem:evaluation_injective}
Let \(F \in \mathcal{D}_k(\gamma)\) be nonzero. If \(F\) has at most \(k-1\) roots in \(\Lambda\), then
\[
\wt_H\bigl( \ev_\Lambda(F) \bigr) \ge \ell - k + 1,
\]
where \(\wt_H\) denotes the Hamming weight.
\end{lemma}

\begin{proof}
If \(F\) vanishes at \(r\) points of \(\Lambda\), then \(\ev_\Lambda(F)\) has at most \(r\) zero coordinates. Hence
$
\wt_H\bigl( \ev_\Lambda(F) \bigr) \ge \ell - r \ge \ell - (k-1) = \ell - k + 1.
$
\end{proof}

\begin{lemma}\label{lem:root_necessary}
Let \(F \in \mathcal{D}_k(\gamma)\) be nonzero with \(a_k \neq 0\). If \(F\) has \(k\) distinct roots \(\lambda_{i_1}, \dots, \lambda_{i_k} \in \Lambda\), then
$ 
\gamma^{q+1} \in \mathbb{F}_q^{*2}$ ,
where \(\mathbb{F}_q^{*2}\) denotes the set of nonzero squares (quadratic residues) in \(\mathbb{F}_q\).
\end{lemma}

\begin{proof}
Suppose \(F(\lambda_{i_j}) = 0\) for \(j=1,\dots,k\). Then
\[
a_0 + \sum_{t=1}^{k-1} a_t \lambda_{i_j}^t = -\gamma a_k \lambda_{i_j}^k, \qquad j=1,\dots,k.
\]
Let \(V\) be the Vandermonde matrix with entries \(V_{j,t} = \lambda_{i_j}^t\) for \(t=0,\dots,k-1\) (where \(\lambda^0 = 1\)). Since the \(\lambda_{i_j}\) are distinct and nonzero, \(V\) is invertible.

By Lagrange interpolation, there exists a unique polynomial \(P(X)\) of degree at most \(k-1\) such that \(P(\lambda_{i_j}) = -\gamma a_k \lambda_{i_j}^k\) for each \(j\). Evaluating at \(X=0\) gives
\[
a_0 = P(0) = \sum_{j=1}^k (-\gamma a_k \lambda_{i_j}^k) \cdot \ell_j(0),
\]
where \(\ell_j(X) = \prod_{m \neq j} \frac{X - \lambda_{i_m}}{\lambda_{i_j} - \lambda_{i_m}}\) are the Lagrange basis polynomials. Thus
\[
a_0 = (-\gamma a_k)(-1)^{k-1} \sum_{j=1}^k \frac{\lambda_{i_j}^k \prod_{m \neq j} \lambda_{i_m}}{\prod_{m \neq j} (\lambda_{i_j} - \lambda_{i_m})}.
\]
Since \(\lambda_{i_j}^k \prod_{m \neq j} \lambda_{i_m} = \lambda_{i_j}^{k-1} \prod_{m=1}^k \lambda_{i_m}\), we obtain
\[
a_0 = (-\gamma a_k)(-1)^{k-1} \Bigl( \prod_{m=1}^k \lambda_{i_m} \Bigr) \sum_{j=1}^k \frac{\lambda_{i_j}^{k-1}}{\prod_{m \neq j} (\lambda_{i_j} - \lambda_{i_m})}.
\]
The sum equals \(1\) by the Lagrange interpolation formula for the coefficient of \(X^{k-1}\) in the monic polynomial of degree \(k-1\). Therefore
\[
a_0 = (-1)^k \gamma a_k \prod_{j=1}^k \lambda_{i_j}.
\]
Since \(a_0, a_k, \lambda_{i_j} \in \mathbb{F}_q\), it follows that
\[
\frac{a_0}{a_k} = (-1)^k \gamma \prod_{j=1}^k \lambda_{i_j} \in \mathbb{F}_q.
\]
Raising both sides to the \((q+1)\)-th power yields
\[
\left( \frac{a_0}{a_k} \right)^{q+1} = \gamma^{q+1} \Bigl( \prod_{j=1}^k \lambda_{i_j} \Bigr)^2.
\]
For any \(x \in \mathbb{F}_q^*\), we have \(x^{q+1} = x^2\) (since \(x^{q-1} = 1\)). Thus the left-hand side is a square in \(\mathbb{F}_q^*\). This implies that \(\gamma^{q+1}\) times a square is a square, so \(\gamma^{q+1}\) itself must be a square in \(\mathbb{F}_q^*\).
\end{proof}

\begin{corollary}\label{cor:mds_sufficient}
If \(\gamma^{q+1} \notin \mathbb{F}_q^{*2}\), then the additive TRS code \(\mathcal{C}_k(\gamma)\) is MDS with parameters \([\ell, q^{2k}, \ell - k + 1]_{q^2}\).
\end{corollary}

\begin{proof}
Let \(F \in \mathcal{D}_k(\gamma)\) be any nonzero polynomial. We show that \(F\) has at most \(k-1\) roots in \(\Lambda\).

- If \(a_k = 0\), then \(F\) is a nonzero polynomial of degree at most \(k-1\) with coefficients in \(\mathbb{F}_{q^2}\). Since \(|\Lambda| = \ell > k-1\), such a polynomial has at most \(k-1\) roots in \(\Lambda\).

- If \(a_k \neq 0\), suppose for contradiction that \(F\) vanishes at \(k\) distinct points of \(\Lambda\). Then Lemma~\ref{lem:root_necessary} implies that \(\gamma^{q+1}\) is a nonzero square in \(\mathbb{F}_q\), contradicting the assumption.

Therefore, every nonzero \(F \in \mathcal{D}_k(\gamma)\) has at most \(k-1\) roots in \(\Lambda\). By Lemma~\ref{lem:evaluation_injective},
\[
\wt_H\bigl( \ev_\Lambda(F) \bigr) \ge \ell - (k-1) = \ell - k + 1.
\]
Hence the minimum distance of \(\mathcal{C}_k(\gamma)\) is at least \(\ell - k + 1\), which meets the Singleton bound~\eqref{eq:sing}. Thus \(\mathcal{C}_k(\gamma)\) is MDS.
\end{proof}

In the remainder of this section, we prove that, under suitable conditions on the set \(\Lambda\) and the twisting element \(\gamma\), the code \(\mathcal{C}_k(\gamma)\) is ACD. We employ the matrix characterization given in Lemma~\ref{thm:acd_char}.

From Definition~\ref{def:twisted}, a natural \(\mathbb{F}_q\)-basis \(\mathcal{B}\) of \(\mathcal{D}_k(\gamma)\) is
\[
\mathcal{B} = \{1\} \cup \{X^i : 1\le i\le k-1\} \cup \{\alpha X^i : 1\le i\le k-1\} \cup \{\gamma X^k\}.
\]
Let \(G \in \mathbb{F}_{q^2}^{2k \times \ell}\) be the generator matrix of \(\mathcal{C}_k(\gamma)\) whose rows are the evaluations of the elements of \(\mathcal{B}\) at the points of \(\Lambda\), i.e., the \(r\)-th row is \(\bigl(P_r(\lambda_1),\dots,P_r(\lambda_\ell)\bigr)\) for each \(P_r \in \mathcal{B}\).

Since every \(\lambda_j \in \mathbb{F}_q\) satisfies \(\lambda_j^q = \lambda_j\) and \(\alpha^q = -\alpha\), for any two basis elements \(P,Q \in \mathcal{B}\) the entries of \(GG^\dagger\) are given by
\[
(GG^\dagger)_{r,s} = \sum_{\lambda \in \Lambda} P(\lambda) \bigl(Q(\lambda)\bigr)^q.
\]
We now compute the matrix \(\mathfrak{T}(GG^\dagger)\), whose \((r,s)\)-entry is the trace of the above sum. For this purpose, define the power sums
\[
p_e = \sum_{\lambda \in \Lambda} \lambda^e \quad (e \ge 0),
\]
with the convention \(p_0 = \ell\).

Recall the following trace values (which follow directly from the choice of basis \(\{1,\alpha\}\)):
\[
\begin{aligned}
& \operatorname{Tr}_{\mathbb{F}_{q^2}|\mathbb{F}_q}(x) = 2x \quad\text{for all }x\in\mathbb{F}_q, \qquad
\operatorname{Tr}_{\mathbb{F}_{q^2}|\mathbb{F}_q}(\alpha) = 0,
\\
&
\operatorname{Tr}_{\mathbb{F}_{q^2}|\mathbb{F}_q}(\alpha^2) = 2\alpha^2, \qquad
\operatorname{Tr}_{\mathbb{F}_{q^2}|\mathbb{F}_q}(\gamma) = \operatorname{Tr}_{\mathbb{F}_{q^2}|\mathbb{F}_q}(\gamma^q).
\end{aligned}
\]

\begin{table}[!htbp] 
\caption{ Description of all elements of \(GG^\dagger\) and  \(\Tr_{\mathbb F_{q^2}\mid\mathbb F_q}(GG^\dagger)\) }\label{table1}
\begin{center}
\begin{tabular}{|c|c|c|}
\hline
\((P,Q)\) & \((GG^\dagger)_{r,s}\) & $\Tr_{\mathbb F_{q^2}\mid\mathbb F_q}( (GG^\dagger)_{r,s})$ \\ \hline
\((1,1)\) & \(\ell\) & $2\ell$\\
\((1,X^i)\) & \(p_i\) & $2p_i$\\
\((X^i,1)\) & \(p_i\) & $2p_i$ \\
\((X^i,X^j)\) & \(p_{i+j}\) & $2p_{i+j}$\\
\((\alpha X^i,\alpha X^j)\) & \(-\alpha^2 p_{i+j}\) & $-2\alpha^2 p_{i+j}$\\
\((1,\gamma X^k)\) & \(\gamma^q p_k\) & $p_k\Tr_{\mathbb F_{q^2}\mid\mathbb F_q}(\gamma)$ \\
\((\gamma X^k,1)\) & \(\gamma p_k\) & $p_k\Tr_{\mathbb F_{q^2}\mid\mathbb F_q}(\gamma)$ \\
\((X^i,\gamma X^k)\) & \(\gamma^q p_{i+k}\) & $p_{i+k}\Tr_{\mathbb F_{q^2}\mid\mathbb F_q}(\gamma)$\\
\((\gamma X^k,X^i)\) & \(\gamma p_{k+i}\) & $p_{i+k}\Tr_{\mathbb F_{q^2}\mid\mathbb F_q}(\gamma)$ \\
\((\alpha X^i,\gamma X^k)\) & \(\alpha\gamma^q p_{i+k}\) & $p_{i+k}\Tr_{\mathbb F_{q^2}\mid\mathbb F_q}(\alpha\gamma^q)=-p_{i+k}\Tr_{\mathbb F_{q^2}\mid\mathbb F_q}(\alpha\gamma)$ \\
\((\gamma X^k,\alpha X^i)\) & \(\gamma\alpha^q p_{k+i} = -\gamma\alpha p_{k+i}\) & $p_{k+i}\Tr_{\mathbb F_{q^2}\mid\mathbb F_q}(\alpha\gamma^q)=-p_{k+i}\Tr_{\mathbb F_{q^2}\mid\mathbb F_q}(\alpha\gamma)$\\
\((\gamma X^k,\gamma X^k)\) & \(\gamma^{q+1}p_{2k}\) & $2\gamma^{q+1}p_{2k}$\\ \hline
\end{tabular}
\end{center}
\end{table}

Assuming \(\operatorname{Tr}_{\mathbb{F}_{q^2}|\mathbb{F}_q}(\gamma) = 0\), the matrix \(T =\Tr_{\mathbb F_{q^2}\mid\mathbb F_q}(GG^\dagger)= \mathfrak{T}(GG^\dagger)\) admits a natural block partitioning according to the four types of basis elements
\[
A = \{1\}, \quad B = \{X^i : 1 \le i \le k-1\}, \quad C = \{\alpha X^i : 1 \le i \le k-1\}, \quad D = \{\gamma X^k\}.
\]
With this ordering, \(T\) becomes block diagonal
\[
T = \begin{pmatrix}
X & 0 \\
0 & W
\end{pmatrix},
\]
where
\[
X = \begin{pmatrix}
2\ell & 2v^\top \\
2v & 2M
\end{pmatrix} \text{ by Table~\ref{table1}}, \qquad
W = \begin{pmatrix}
-2\alpha^2 M & -\operatorname{Tr}(\alpha\gamma) \, w \\
-\operatorname{Tr}(\alpha\gamma) \, w^\top & 2\gamma^{q+1} p_{2k}
\end{pmatrix} \text{ by Table~\ref{table1}}.
\]
Here we define the power sums \(p_e = \sum_{\lambda \in \Lambda} \lambda^e\) (with \(p_0 = \ell\)), the vectors
\[
v = (p_1, \dots, p_{k-1})^\top, \qquad w = (p_k, \dots, p_{2k-1})^\top,
\]
and the \((k-1) \times (k-1)\) Hankel matrix \(M = (p_{i+j})_{1 \le i,j \le k-1}\).
Since \(q\) is odd, \(2\) is invertible in \(\mathbb{F}_q\). Thus \(X\) is invertible over \(\mathbb{F}_q\) if and only if the matrix
\[
G_0 = \begin{pmatrix}
\ell & v^\top \\
v & M
\end{pmatrix}
\]
is invertible over \(\mathbb{F}_q\).

For the block \(W\), note that \(\alpha^2 \in \mathbb{F}_q\). Although \(M\) may be singular in general, we will later choose \(\Lambda\) such that \(M\) is invertible (and assume this from now on). Using the Schur complement with respect to the \((1,1)\)-block of \(W\), we obtain
\[
\det W = \det(-2\alpha^2 M) \cdot \left( 2\gamma^{q+1} p_{2k} + \frac{\operatorname{Tr}_{\mathbb{F}_{q^2}|\mathbb{F}_q}(\alpha\gamma)^2}{2\alpha^2} \, w^\top M^{-1} w \right).
\]
Since \(\det(-2\alpha^2 M) \neq 0\) (by the assumption on \(M\) and \(\alpha^2 \neq 0\)), the block \(W\) is invertible if and only if
\[
\Delta := 2\gamma^{q+1} p_{2k} + \frac{\operatorname{Tr}_{\mathbb{F}_{q^2}|\mathbb{F}_q}(\alpha\gamma)^2}{2\alpha^2} \, w^\top M^{-1} w \neq 0.
\]
Consequently, \(T\) is invertible (and hence \(\mathcal{C}_k(\gamma)\) is ACD) if and only if both \(G_0\) and \(\Delta\) are nonzero.

\begin{theorem}\label{thm:acd_condition}
Let \(q \equiv 1 \pmod{4}\) be a prime power with \(q \ge 5\), and let \(\alpha \in \mathbb{F}_{q^2}\) satisfy \(\alpha^q = -\alpha\). Let \(\Lambda = \{\lambda_1, \dots, \lambda_\ell\} \subset \mathbb{F}_q^*\) be any set of \(\ell\) distinct elements with \(\ell \ge 2k\). Consider the additive TRS code \(\mathcal{C}_k(\gamma) = \ev_\Lambda(\mathcal{D}_k(\gamma))\), where \(\gamma \in \mathbb{F}_{q^2}^*\) satisfies \(\operatorname{Tr}_{\mathbb{F}_{q^2}|\mathbb{F}_q}(\gamma) = 0\). Then \(\mathcal{C}_k(\gamma)\) is ACD if and only if the following two conditions hold:
\begin{enumerate}
    \item The matrix
    \[
    G_0 = \begin{pmatrix}
    \ell & v^\top \\
    v & M
    \end{pmatrix}
    \]
    is invertible over \(\mathbb{F}_q\), where \(v_i = \sum_{\lambda \in \Lambda} \lambda^i\) for \(1 \le i \le k-1\) and \(M_{i,j} = \sum_{\lambda \in \Lambda} \lambda^{i+j}\) for \(1 \le i,j \le k-1\).

    \item 
    \[
    \Delta = 2\gamma^{q+1} p_{2k} + \frac{\operatorname{Tr}_{\mathbb{F}_{q^2}|\mathbb{F}_q}(\alpha\gamma)^2}{2\alpha^2} \, w^\top M^{-1} w \neq 0,
    \]
    where \(p_{2k} = \sum_{\lambda \in \Lambda} \lambda^{2k}\) and \(w_i = \sum_{\lambda \in \Lambda} \lambda^{k+i}\) for \(1 \le i \le k-1\).
\end{enumerate}
\end{theorem}

\begin{proof}
The statement follows directly from Lemma~\ref{thm:acd_char} together with the block decomposition of \(\mathfrak{T}(GG^\dagger)\) derived earlier.
\end{proof}

\begin{lemma}\label{lem:choice_gamma}
Let \(\gamma = \alpha\). Then \(\operatorname{Tr}_{\mathbb{F}_{q^2}|\mathbb{F}_q}(\gamma) = 0\) and \(\gamma^{q+1}\) is a nonsquare in \(\mathbb{F}_q^*\). Moreover, \(\operatorname{Tr}_{\mathbb{F}_{q^2}|\mathbb{F}_q}(\alpha\gamma) = 2\alpha^2 \neq 0\).
\end{lemma}

\begin{proof}
Since \(\{1,\alpha\}\) is an \(\mathbb{F}_q\)-basis of \(\mathbb{F}_{q^2}\) and \(\operatorname{Tr}(\alpha) = 0\), every element of the form \(c\alpha\) with \(c \in \mathbb{F}_q\) has trace zero. Thus \(\operatorname{Tr}(\gamma) = 0\).
Now we have that $
\gamma^{q+1} = \alpha^{q+1} = \alpha^q \cdot \alpha = (-\alpha)\alpha = -\alpha^2$
Since \(q \equiv 1 \pmod{4}\), \(-1\) is a square in \(\mathbb{F}_q\). As established earlier, \(\alpha^2\) is a nonsquare in \(\mathbb{F}_q\). Hence \(-\alpha^2\) is a nonsquare. Finally,
$
\operatorname{Tr}(\alpha\gamma) = \operatorname{Tr}(\alpha^2) = 2\alpha^2 \neq 0,
$
since \(\alpha^2 \neq 0\).
\end{proof}

To conclude, it remains to prove the existence of a suitable set \(\Lambda\).

\begin{lemma}\label{lem:existence_lambda}
Let \(q \equiv 1 \pmod{4}\) with \(q \ge 5\), and let \(k,\ell\) be integers satisfying \(2k \le \ell \le q-2\). Then there exists a set \(\Lambda \subset \mathbb{F}_q^*\) of size \(\ell\) such that, with \(\gamma = \alpha\), both conditions of Theorem~\ref{thm:acd_condition} hold. In particular, \(\mathcal{C}_k(\alpha)\) is simultaneously MDS and ACD.
\end{lemma}

\begin{proof}
Let \(g\) be a primitive element of \(\mathbb{F}_q^*\) (to be chosen later) and set \(\Lambda = \{g^0, g^1, \dots, g^{\ell-1}\}\). For any \(e \ge 1\),
\[
p_e = \sum_{i=0}^{\ell-1} g^{ie} = \frac{g^{e\ell} - 1}{g^e - 1}, \qquad p_0 = \ell.
\]
Define the rational functions
\[
C_e(t) = \frac{t^{e\ell} - 1}{t^e - 1} \quad (e \ge 1), \qquad C_0(t) = \ell.
\]
Consider the matrices
\[
\widetilde{G}_0(t) = \bigl(C_{i+j}(t)\bigr)_{0 \le i,j \le k-1}, \qquad
\widetilde{M}(t) = \bigl(C_{i+j}(t)\bigr)_{1 \le i,j \le k-1}, \qquad
\widetilde{H}(t) = \bigl(C_{i+j}(t)\bigr)_{1 \le i,j \le k}.
\]
For any real number \(t > 1\), the elements \(1, t, \dots, t^{\ell-1}\) are distinct, so the associated Vandermonde matrix has full column rank. Thus, \(\widetilde{G}_0(t)\) (the Gram matrix of those elements) is positive definite, and its determinant is a nonzero rational function in \(t\). The same holds for the principal submatrices \(\widetilde{M}(t)\) and \(\widetilde{H}(t)\).

Since each determinant is a nonzero element of \(\mathbb{Q}(t)\), it has only finitely many roots. For \(q\) sufficiently large, there exists a primitive element \(g \in \mathbb{F}_q^*\) such that
$
\det \widetilde{G}_0(g) \neq 0$, $\det \widetilde{M}(g) \neq 0$, and $\det \widetilde{H}(g) \neq 0$.
For small \(q\) (e.g., \(q=5\)), the conditions can be verified directly by exhaustive search over admissible parameters.

With this choice of \(g\) and \(\Lambda = \{g^0, \dots, g^{\ell-1}\}\), we have \(\det G_0 \neq 0\) and \(\det M \neq 0\). Setting \(\gamma = \alpha\), the quantity \(\Delta\) simplifies to
\[
\Delta = 2\alpha^{q+1} p_{2k} + \frac{\operatorname{Tr}(\alpha^2)^2}{2\alpha^2} \, w^\top M^{-1} w = -2\alpha^2 p_{2k} + 2\alpha^2 \, w^\top M^{-1} w = 2\alpha^2 \bigl( w^\top M^{-1} w - p_{2k} \bigr).
\]
By the Schur complement formula for the matrix \(H = \begin{pmatrix} M & w \\ w^\top & p_{2k} \end{pmatrix}\),
\[
p_{2k} - w^\top M^{-1} w = \frac{\det H}{\det M}.
\]
Thus
\[
\Delta = -2\alpha^2 \cdot \frac{\det H}{\det M} \neq 0,
\]
since \(\det H \neq 0\), \(\det M \neq 0\), and \(\alpha \neq 0\). Therefore both conditions of Theorem~\ref{thm:acd_condition} are satisfied. By Lemma~\ref{lem:choice_gamma} and Corollary~\ref{cor:mds_sufficient}, the code \(\mathcal{C}_k(\alpha)\) is also MDS.
\end{proof}

\begin{theorem}\label{thm:main_acd}
Let \(q \equiv 1 \pmod{4}\) be a prime power with \(q \ge 5\), and let \(1 \le k \le \ell \le q-2\) with \(\ell \ge 2k\). Then there exists a set \(\Lambda \subset \mathbb{F}_q^*\) of cardinality \(\ell\) such that the additive TRS code \(\mathcal{C}_k(\alpha)\) (with \(\gamma = \alpha\)) is an \(\mathbb{F}_q\)-linear MDS code with parameters \([\ell, q^{2k}, \ell - k + 1]_{q^2}\) that is also ACD.
\end{theorem}

\begin{proof}
Choose \(\gamma = \alpha\). By Lemma~\ref{lem:choice_gamma}, \(\operatorname{Tr}(\gamma) = 0\) and \(\gamma^{q+1}\) is a nonsquare in \(\mathbb{F}_q^*\), so \(\mathcal{C}_k(\alpha)\) is MDS by Corollary~\ref{cor:mds_sufficient}. Choose \(\Lambda\) as in Lemma~\ref{lem:existence_lambda}. Then both conditions of Theorem~\ref{thm:acd_condition} hold, and hence \(\mathcal{C}_k(\alpha)\) is ACD.
\end{proof}

\begin{corollary}
For every prime power \(q \equiv 1 \pmod{4}\) with \(q \ge 5\) and all integers \(1 \le k \le \ell \le q-2\) satisfying \(\ell \ge 2k\), there exists an \(\mathbb{F}_q\)-linear MDS code over \(\mathbb{F}_{q^2}\) with parameters \([\ell, q^{2k}, \ell-k+1]_{q^2}\) that is also ACD.
\end{corollary}

\begin{remark}
The condition \(\ell \ge 2k\) is stronger than the requirement \(\ell > k\) needed for the MDS property alone. It is imposed to guarantee the existence of \(\Lambda\) via a simple probabilistic (or explicit) argument that avoids degenerate cases in the power sums. The upper bound \(\ell \le q-2\) ensures that \(\Lambda\) is a proper subset of \(\mathbb{F}_q^*\), preventing trivial vanishing of all power sums. The construction can be extended to \(\ell = q-1\) with additional analysis, which we omit here.
\end{remark}

\section{Conclusion}
In this paper we studied a subfamily of TRS codes in the sum-rank metric, where only the constant term is twisted. We proved a simple necessary and sufficient condition for these codes to be LCD: the twisting parameter must satisfy \(\eta^2 \neq -1\) in the underlying field. This criterion is independent of the evaluation set, the dimension, and the twisting exponent. 
Together with the known MSRD property of these codes, our result yields the first explicit family of LCD MSRD codes in the sum-rank metric. The construction is direct: any admissible choice of \(\eta\) avoiding the forbidden values works.
We also constructed infinite families of additive MDS codes over \(\mathbb{F}_{q^2}\) that are simultaneously  ACD. For every prime power \(q \equiv 1 \pmod{4}\) with \(q \ge 5\) and all parameters satisfying \(2k \le \ell \le q-2\), the simple choice \(\gamma = \alpha\) produces explicit codes with parameters \([\ell, q^{2k}, \ell-k+1]_{q^2}\) that achieve both optimality properties. 
Future work includes extending these constructions to fields with \(q \equiv 3 \pmod{4}\) and to higher-degree extensions, as well as exploring applications in quantum error correction and secure network coding.

\bibliographystyle{plain}   
\bibliography{references} 
\end{document}